\documentclass[11pt,fleqn]{article}

\usepackage{amsmath,amssymb,amsthm,enumerate} 

\newcommand{\p}{\partial}

\newcommand{\rank}{\mathop{\rm rank}\nolimits}
\newcommand{\sign}{\mathop{\rm sign}\nolimits}
\newcommand{\CL}{\mathop{\rm CL}\nolimits}
\newcommand{\ord}{\mathop{\rm ord}\nolimits}

\newcommand{\todo}[1][\null]{\ensuremath{\clubsuit}}

\newcommand{\noprint}[1]{}

\newtheorem{theorem}{Theorem}
\newtheorem{lemma}{Lemma}
\newtheorem{corollary}{Corollary}
\newtheorem{proposition}{Proposition}
{\theoremstyle{definition} 
\newtheorem{example}{Example}
\newtheorem{remark}{Remark}
}

\begin{document}

\par\noindent {\LARGE\bf
Conservation Laws and Normal Forms\\ of Evolution Equations\par}

{\vspace{4mm}\par\noindent Roman O. POPOVYCH~$^\dag$ and Artur
SERGYEYEV~$^\ddag$
\par\vspace{2mm}\par}

{\vspace{2mm}\par\noindent\it $^{\dag}$~Fakult\"at f\"ur Mathematik, Universit\"at Wien, Nordbergstra{\ss}e 15, A-1090 Wien, Austria \par}

{\par\noindent \it $\phantom{^{\dag}}$~Institute of Mathematics of NAS of Ukraine, 3 Tereshchenkivska Str., Kyiv-4, Ukraine \par}

{\vspace{2mm}\par\noindent \it $^{\ddag}$~Mathematical Institute, Silesian University in Opava, Na Rybn\'\i{}\v{c}ku 1, 746 01 Opava,\\
$\phantom{^{\dag}}$~Czech Republic \par}

{\vspace{2mm}\par\noindent $\phantom{^{\dag,\ddag}}$\rm E-mail: \it  $^\dag$rop@imath.kiev.ua, $^\ddag$Artur.Sergyeyev@math.slu.cz
 \par}

\vspace{-3ex}



\begin{abstract}
We study local conservation laws for evolution
equations in two independent variables. In particular, we present
normal forms for the equations admitting one or two low-order
conservation laws. Examples include Harry Dym equation,
Korteweg--de-Vries-type equations, and Schwarzian KdV equation. It is also shown that for linear
evolution equations all their conservation laws are (modulo trivial conserved vectors)
at most quadratic in the dependent variable and its derivatives.
\end{abstract}

\vspace{-3ex}

\section{Introduction}

The role played in the sciences by linear and nonlinear evolution
equations and, in particular, by conservation laws thereof, is hard
to overestimate (recall e.g.\ linear and nonlinear Schr\"odinger
equations and the Korteweg--de Vries (KdV) equation in physics,
reaction-diffusion systems in chemistry and biology, and the
Black--Scholes equation in the finance, to name just a few). For
instance, the discovery of higher conservation laws for the
KdV equations provided an important milestone on the way that has eventually lead to the
discovery of the inverse scattering transform and the modern theory
of integrable systems, see e.g.\ \cite{Miura&Gardner&Kruskal1968, Newell1985}. 
However, the theory of conservation laws for
evolution equations is still far from being complete even for the
simplest case of two independent variables, and in the present
paper we address some issues of the theory in question for this very
case.

We shall deal with an evolution equation in two independent variables,
\begin{equation}\label{EqGenEvol}
u_t=F(t,x,u_0,u_1,\dots,u_n),\quad n\ge 2,\quad F_{u_n}\neq 0,
\end{equation}
where $u_j\equiv \p^j u/\p x^j$, $u_0\equiv u$, and $F_{u_j}=\p F/\p u_j$.
We shall also employ, depending on convenience or necessity, 
the following notation for low-order derivatives: 
$u_x=u_1$, $u_{xx}=u_2$, and $u_{xxx}=u_3$.

There is a considerable body of results on conservation laws of
evolution equations of the form~\eqref{EqGenEvol}. 
For instance, in the seminal paper~\cite{Bryant&Griffiths1995b} the authors studied,
{\em inter alia}, conservation laws of Eq.~\eqref{EqGenEvol} with $\p
F/\p t=0$ for $n=2$. They proved that the possible dimensions of
spaces of inequivalent conservation laws for such equations are 0,
1, 2 and $\infty$, and described the equations possessing spaces of
conservation laws of these dimensions (the precise definitions of
equivalence and order of conservation laws are given in the next
section). These results were further generalized in
\cite{Popovych&Samoilenko2008} for the case when $F$ explicitly depends on $t$.

Important results on conservation laws of \eqref{EqGenEvol},
typically under the assumptions of polynomiality and
$t,x$-independence of $F$ and of the conservation laws themselves,
were obtained in \cite{ag1, ag2, ag3, ag4, fol1, fol2, gal, kap}.
However, for general Eq.~\eqref{EqGenEvol} there is
no simple picture analogous to that of the second-order case
discussed above. For instance, unlike the second-order case, there
exist odd-order evolution equations that possess infinitely many
inequivalent conservation laws of increasing orders without being
linearizable. Rather, such equations are integrable via the inverse
scattering transform, the famous KdV equation providing a prime
example of such behavior, see e.g.\ \cite{Ibragimov1985, kap, Miura&Gardner&Kruskal1968}
and references therein; for the fifth-order equations see
\cite{dss}.

Note that many results on symmetries and conservation laws were
obtained using the formal symmetry approach and modifications
thereof, see e.g.\ the recent survey \cite{miksok} and references
therein, in particular \cite{miksok0,svso}. For instance, it was shown that
an equation~\eqref{EqGenEvol} of even order ($n=2m$) has no
conservation laws (modulo trivial ones) of order greater than $m$,
see \cite{ag1, fol1, Ibragimov1985, ig} for details. There also
exists a closely related approach to the study of symmetries and
conservation laws of evolution equations, the so-called symbolic
method, see \cite{mnw, sw1, sw2, sw3} and references therein for details.

However, many important questions concerning the conservation laws of
evolution equations were not answered so far. For example, we are not
aware of any significant advances in the study of normal forms of
evolution equations admitting low-order conservation laws considered
in \cite{Bryant&Griffiths1995b, fol2, Popovych&Samoilenko2008}. In
the present paper we provide such normal forms with respect to
contact or point transformations for equations
admitting one or two low-order conservation laws, respectively,
see Theorems~\ref{TheoremOnOneConsLawOfEvolEqs} and \ref{TheoremOnTwoConsLawOfEvolEqs} below.
Let us stress that in what follows we restrict ourselves to considering only
{\em local} conservation laws whose densities and fluxes
depend only on the independent and dependent variables and a finite number
of the derivatives of the latter.

The complete description of conservation laws for {\em linear}
evolution equations with $t,x$-dependent coefficients was also missing so far.
Below we show that linear even-order equations of the form~\eqref{EqGenEvol}
can only possess conservation laws linear in
$u_j$ for all $j=0,1,2,\dots$ while the odd-order equations can further
admit the conservation laws (at most) quadratic in $u_j$, see
Theorems~\ref{lin-cl-th1} and \ref{lin-cl-th2},
Corollary~\ref{peolin} and Theorem~\ref{poolin} below. This
naturally generalizes some earlier results from \cite{ag3,gal}; cf.\ also \cite{ashton}.
The generation of linear and quadratic conservation laws for linear differential equations
is also discussed in some depth in \cite[Section~5.3]{Olver1993}.

Below we denote by $\CL(\mathcal E)$ the space of local conservation
laws of~$\mathcal E$ (cf.\ Section~\ref{SectionOnAuxiliaryStatements}), where
$\mathcal E$ denotes a fixed equation from the class~\eqref{EqGenEvol}.
In what follows $D_t$ and $D_x$ stand for
the total derivatives (see e.g.\ \cite{Olver1993} for details)
with respect to the variables~$t$ and~$x$,
\[
D_t=\p_t+u_t\p_u+u_{tt}\p_{u_t}+u_{tx}\p_{u_x}+\cdots,\quad
D_x=\p_x+u_x\p_u+u_{tx}\p_{u_t}+u_{xx}\p_{u_x}+\cdots.
\]
As usual, the subscripts like $t$, $x$, $u$, $u_x$, etc.\ stand for the partial
derivatives in the respective variables.

\section{Admissible transformations of evolution equations}%
\label{SectionOnAdmissibleContactTransformationsOfEvolutionEquations}

The contact transformations mapping an equation from class~\eqref{EqGenEvol} into
another equation from the same class are well known \cite{Magadeev1993}
to have the form
\begin{equation}\label{EqContactTransOfGenEvolEqs}
\tilde t=T(t), \quad
\tilde x=X(t,x,u,u_x), \quad
\tilde u=U(t,x,u,u_x).
\end{equation}
The functions~$T$, $X$ and $U$ must satisfy the nondegeneracy assumptions, namely, $T_t\ne0$ and
\begin{equation}\label{EqNondegeneracyAssumptionForContactTransOfGenEvolEqs}
\rank\left(\begin{array}{ccc}X_x&X_u&X_{u_x}\\U_x&U_u&U_{u_x}\end{array}\right)=2,
\end{equation}
and the contact condition
\begin{equation}\label{EqContactConditionForContactTransOfGenEvolEqs}
(U_x+U_uu_x)X_{u_x}=(X_x+X_uu_x)U_{u_x}.
\end{equation}
The transformation~\eqref{EqContactTransOfGenEvolEqs} is uniquely
extended to the derivative $u_x$ and to the higher derivatives by the formulas
$\tilde u_{\tilde x}=V(t,x,u,u_x)$ and $\tilde u_k\equiv\p^k \tilde u/\p \tilde x^k=((1/D_xX)D_x)^kV$,
where
\[
V=\frac{U_x+U_uu_x}{X_x+X_uu_x}\quad\mbox{or}\quad V=\frac{U_{u_x}}{X_{u_x}}
\]
if $X_x+X_uu_x\ne0$ or $X_{u_x}\ne0$, respectively; the possibility of simultaneous vanishing
of these two quantities is ruled out by (\ref{EqNondegeneracyAssumptionForContactTransOfGenEvolEqs}).

{\samepage
The transformed equation \eqref{EqGenEvol} reads $\tilde u_{\tilde t}=\tilde F$ where
\begin{equation}\label{EqTransRightHandSideOfGenEvolEqs}
\tilde F=\frac{U_u-X_uV}{T_t}F+\frac{U_t-X_tV}{T_t},
\end{equation}
and $(X_u,U_u)\ne(0,0)$ because of~\eqref{EqNondegeneracyAssumptionForContactTransOfGenEvolEqs}
and~\eqref{EqContactConditionForContactTransOfGenEvolEqs}.

}

Any transformation of the
form~\eqref{EqContactTransOfGenEvolEqs} leaves the
class~\eqref{EqGenEvol} invariant, and therefore its extension to an
arbitrary element~$F$ belongs to the contact equivalence
group~$G^\sim_{\rm c}$ of class~\eqref{EqGenEvol}, so there are no
other elements in~$G^\sim_{\rm c}$. In other words, the equivalence
group~$G^\sim_{\rm c}$ generates the whole set of admissible contact
transformations in the class~\eqref{EqGenEvol}, i.e., this class is
normalized with respect to contact transformations, see
\cite{Popovych&Kunzinger&Eshraghi2006} for details.

The above results can be summarized as follows.
\begin{proposition}
The class of equations \eqref{EqGenEvol} is contact-normalized.
The contact equivalence group~$G^\sim_{\rm c}$ of the class~\eqref{EqGenEvol} is formed by the transformations~\eqref{EqContactTransOfGenEvolEqs},
satisfying conditions~\eqref{EqNondegeneracyAssumptionForContactTransOfGenEvolEqs} and~\eqref{EqContactConditionForContactTransOfGenEvolEqs}
and prolonged to the arbitrary element~$F$ by~\eqref{EqTransRightHandSideOfGenEvolEqs}.
\end{proposition}

Furthermore, the class~\eqref{EqGenEvol} is also point-normalized.
The point equivalence group~$G^\sim_{\rm p}$ of this class consists of the transformations
of the form
\begin{equation}\label{EqPointTransOfGenEvolEqs}
\tilde t=T(t),\quad \tilde x=X(t,x,u),\quad \tilde u=U(t,x,u),\quad
\tilde F=\frac\Delta{T_tD_xX}F+\frac{U_tD_xX-X_tD_xU}{T_tD_xX},
\end{equation}
where $T$, $X$ and $U$ are arbitrary smooth functions that satisfy the nondegeneracy conditions
$T_t\ne0$ and $\Delta=X_xU_u-X_uU_x\ne0$.

Notice that there exist subclasses of the class~\eqref{EqGenEvol} whose
sets of admissible contact transformations are {\em exhausted} by
point transformations.

In the present paper we do  not consider more general transformations, e.g., differential substitutions
such as the Cole--Hopf transformation.

\section{Some basic results on conservation laws}%
\label{SectionOnAuxiliaryStatements}

It is well known that for any evolution equation
\eqref{EqGenEvol} we can assume without loss of generality  that
the associated quantities like symmetries, cosymmetries, densities,
etc., can be taken to be independent of the $t$-derivatives or
mixed derivatives of $u$. We shall stick to this assumption
throughout the rest of the paper.

Following \cite{Olver1993} we shall refer to a (smooth) function of
$t$, $x$ and a finite number of $u_j$ as to a {\em differential
function}. Given a differential function $f$, its {\em order}
(denoted by $\ord f$) is the greatest integer $k$ such
that $f_{u_k}\neq 0$ but $f_{u_j}=0$ for all $j>k$.
For $f=f(t,x)$ we assume that $\ord f=0$.

Thus, for a (fixed) evolution equation \eqref{EqGenEvol}, which we denote
by $\mathcal E$ as before, we lose no generality \cite{Olver1993} in
considering only the {\em conserved vectors} of the form $(\rho,\sigma)$,
where $\rho$ and $\sigma$ are differential functions which satisfy the condition
\begin{equation}\label{cl0}
D_t\rho+D_x\sigma=0 \bmod \check{\mathcal E},
\end{equation}
and $\check{\mathcal E}$ means the equation $\mathcal E$ together with all its differential consequences.
Here $\rho$ is the {\em density} and $\sigma$ is the {\em flux} for the conserved vector $(\rho,\sigma)$.
Let
\[
\frac{\delta}{\delta u}=\sum\limits_{i=0}^{\infty}(-D_x)^i\p_{u_i}, \quad
f_*=\sum\limits_{i=0}^{\infty}f_{u_i}D_x^i, \quad
f_*^\dagger=\sum\limits_{i=0}^{\infty}(-D_x)^i\circ f_{u_i},
\]
denote the operator of variational derivative, the Fr\'echet derivative of a differential function~$f$,
and its formal adjoint, respectively.
With this notation in mind we readily infer that the condition~\eqref{cl0} can be rewritten as $\rho_t+\rho_* F+D_x\sigma=0$.
As $\rho_* F=F\delta\rho/\delta u+D_x\zeta$ for some differential function~$\zeta$,
see e.g.\ \cite[Section~22.5]{Ibragimov1985},
there exists a differential function~$\Psi$ (in fact, $\Psi=-\zeta-\sigma$) such that
\begin{equation}\label{cl0a}
\rho_t+F\frac{\delta\rho}{\delta u}=D_x\Psi.
\end{equation}

A conserved vector $(\rho,\sigma)$ is called \emph{trivial} if it satisfies the condition $D_t\rho+D_x\sigma=0$ on the entire jet space.
It is easily seen that the conserved vector $(\rho,\sigma)$ is trivial
if and only if $\rho\in\mathop{\rm Im} D_x$, i.e., there exists a
differential function $\zeta$ such that $\rho=D_x\zeta$.
Two conserved vectors are {\em equivalent} if they differ by a trivial conserved vector.
We shall call a {\em conservation law} of~$\mathcal E$
an equivalence class of conserved vectors of~$\mathcal E$.
The set $\CL(\mathcal E)$ of conservation laws of~$\mathcal E$ is a vector space,
and the zero element of this space is the conservation law being the equivalence class of
trivial conserved vectors. This is why nonzero conservation laws are usually called {\em nontrivial}.

For any conservation law~$\mathcal L$ of~$\mathcal E$ there exists a unique differential function $\gamma$
called the {\em characteristic} of $\mathcal L$ such that
for any conserved vector $(\rho,\sigma)$ associated with $\mathcal L$ (we shall write this as $(\rho,\sigma)\in \mathcal L$)
there exists a trivial conserved vector $(\tilde\rho,\tilde\sigma)$ satisfying the condition
\begin{equation}\label{clch}
D_t(\rho+\tilde\rho)+D_x(\sigma+\tilde\sigma)=\gamma (u_t-F).
\end{equation}
It is important to stress that, unlike \eqref{cl0}, the above equation holds
on the entire jet space rather than merely modulo $\check{\mathcal E}$.

The characteristic $\gamma$ of any conservation law satisfies the equation (see e.g.\ \cite{Olver1993})
\begin{equation}\label{cosym}
D_t\gamma+F_*^\dagger\gamma=0\bmod\check{\mathcal E},
\qquad\mbox{or equivalently,}\qquad
\gamma_t+\gamma_*F+F_*^\dagger\gamma=0.
\end{equation}
However, in general a solution of~\eqref{cosym} is not necessarily a characteristic of some conservation
law for~\eqref{EqGenEvol}. Solutions of~\eqref{cosym} are called {\em cosymmetries}, see e.g. \cite{Blaszak}.

It can be shown that the characteristic of the conservation law
associated with a conserved vector of the form $(\rho,\sigma)$
equals $\delta\rho/\delta u$.
This yields a necessary and sufficient condition for
a cosymmetry $\gamma$ to be a characteristic of a conservation law (see e.g.\ \cite{Olver1993}): 
$\gamma_*=\gamma_*^\dagger$.
This condition means simply that the Fr\'echet derivative of $\gamma$ is formally self-adjoint.

The following results are well known, see e.g.\ \cite{fol1} for Lemma~\ref{LemmaOnFormOfConsVecsOfEvolEqs}.

\begin{lemma}\label{LemmaOnFormOfConsVecsOfEvolEqs2}
Suppose that an equation from the class~\eqref{EqGenEvol} admits
a nontrivial conserved vector $(\rho,\sigma)$, where $\ord \rho=k>0$, and $\rho_{u_k u_k}=0$. Then the conserved vector
$(\rho,\sigma)$ is equivalent to a conserved
vector~$(\tilde\rho,\tilde\sigma)$ with $\ord\tilde{\rho}\leqslant k-1$.
\end{lemma}

\begin{proof}
By assumption, $\rho=\rho^1 u_k+\rho^0$, and hence $\sigma=-\rho^1
D_x^{k-1} (F)+\sigma^0$, where
$\rho^1=\rho^1(t,x,u,u_1,\dots,u_{k-1})$,
$\rho^0=\rho^0(t,x,u,u_1,\dots,u_{k-1})$, and
$\sigma^0=\sigma^0(t,x,u,u_1,\dots,u_{k+n})$. Put $\tilde
\rho=\rho-D_x\Phi$ and $\tilde\sigma=\sigma+D_t\Phi$, where
$\Phi=\int\rho^1 du_{k-1}$. Then $\tilde\rho_{u_k}=0$,
$\tilde\sigma_{u_{k+n}}=0$ and $(\tilde\rho,\tilde\sigma)$ is a
conserved vector equivalent to $(\rho,\sigma)$, and $\ord \tilde{\rho}\leqslant k-1$.
\end{proof}

In what follows, for any given conservation law $\mathcal L$ we shall, unless otherwise explicitly stated,
choose a representative (that is, the conserved vector) with the lowest possible order $k$ of the associated density $\rho$.
The order in question (i.e., the greatest integer $k$ such that $\rho_{u_ku_k}\neq 0$ but $\rho_{u_j}=0$ for all $j>k$)
will be called the {\em density order} of $\mathcal L$ and denoted by $\ord_{\rm d}\mathcal L$.
It equals one half of the order of the associated characteristic. 

\begin{lemma}\label{LemmaOnFormOfConsVecsOfEvolEqs}
For any conservation law $\mathcal L$ 
of a $(1+1)$-dimensional even-order ($n=2q$) evolution equation 
of the form \eqref{EqGenEvol} we have $\ord_{\rm d}\mathcal L\leqslant q$. 
\end{lemma}

\section{Evolution equations having low-order conservation laws}%
\label{SectionOnEvolutionEquationsHavingLow-OrderConservationLaws}

Contact and point equivalence transformations can be used
for bringing equations from the class~\eqref{EqGenEvol} that admit (at least) 
one or two nontrivial low-order conservation laws 
into certain special forms. This is achieved through bringing the conservation laws in question to normal forms.

\begin{theorem}\label{TheoremOnOneConsLawOfEvolEqs}
Any pair $(\mathcal E,\mathcal L)$, where $\mathcal E$ is an equation of the form~\eqref{EqGenEvol}
and $\mathcal L$ is a nontrivial conservation law of $\mathcal E$ with $\ord_{\rm d}\mathcal L\leqslant 1$
is $G^\sim_{\rm c}$-equivalent to a pair $(\tilde{\mathcal E},\tilde{\mathcal L})$,
where $\tilde{\mathcal E}$ is an equation of the same form
and $\tilde{\mathcal L}$ is a conservation law of $\tilde{\mathcal E}$ with the characteristic equal to $1$.
\end{theorem}

\begin{proof}
Let $\mathcal T\in G^\sim_{\rm c}$ map an equation $\mathcal E$ into (another) equation~$\tilde{\mathcal E}$
from the same class~\eqref{EqGenEvol}, see Section~\ref{SectionOnAdmissibleContactTransformationsOfEvolutionEquations}.
Quite obviously, the inverse $\mathcal T^{-1}$ of $\mathcal T$
induces (through pullback) a mapping from the space $\CL(\mathcal E)$ of conservation laws of $\mathcal E$
to $\CL(\tilde{\mathcal E})$.
The conserved vectors of~$\mathcal E$ are transformed into those of~$\tilde{\mathcal E}$ according to
the formula
\cite{Popovych&Ivanova2004CLsOfNDCEs,Popovych&Kunzinger&Ivanova2008}
\[
\tilde \rho=\frac \rho{D_xX}, \quad \tilde \sigma=\frac \sigma{T_t}+\frac {D_tX}{D_xX}\frac \rho{T_t}.
\]
Now let an equation~$\mathcal E$ from the class~\eqref{EqGenEvol} have a nontrivial conservation law~$\mathcal L$
with $\ord_{\rm d}\mathcal L\leqslant1$.
Fix a conserved vector $(\rho,\sigma)$ associated with~$\mathcal L$,
and set $T=t$.
The density~$\tilde \rho$ of the transformed conserved vector $(\tilde \rho,\tilde \sigma)$
is easily seen to depend at most on~$\tilde t$, $\tilde x$, $\tilde u$, $\tilde u_{\tilde x}$ and~$\tilde u_{\tilde x\tilde x}$.
Moreover, it is immediate that $\tilde\rho$ is linear in $\tilde u_{\tilde x\tilde x}$,
so we can pass to an equivalent conserved vector $(\bar\rho,\bar\sigma)$
such that $\p\bar\rho/\p\tilde u_{\tilde x\tilde x}=0$, and hence for the transformed counterpart $\tilde{\mathcal L}$ of $\mathcal L$
we have $\ord_{\mathrm{d}} \tilde{\mathcal L}\leqslant1$.

Next, the conservation law $\tilde{\mathcal L}$
associated with $(\tilde \rho,\tilde \sigma)$ has characteristic $1$ if and only if
there exists a function $\tilde\Phi=\tilde\Phi(\tilde t,\tilde x,\tilde u,\tilde u_{\tilde x})$ such that
$\tilde\rho=\tilde u+D_{\tilde x}\tilde\Phi$. Upon going back to the old coordinates $x, t, u, u_x$
and bearing in mind that $\tilde u=U(t,x,u,u_x)$
and $\tilde x=X(t,x,u,u_x)$ this boils down to $D_x\Phi+UD_xX=\rho$,
where  $\Phi(t,x,u,u_x)=\tilde\Phi(\tilde t,\tilde x,\tilde u,\tilde u_{\tilde x})$.
Splitting the equation $D_x\Phi+UD_xX=\rho$
with respect to~$u_{xx}$ yields the system
\begin{equation}\label{EqCLsOfEvolEqsSystemForReductionOf1CL}
\Phi_x+UX_x+(\Phi_u+UX_u)u_x=\rho, \quad \Phi_{u_x}+U X_{u_x}=0.
\end{equation}
This system
in conjunction with the contact condition~\eqref{EqContactConditionForContactTransOfGenEvolEqs}
has, {\em inter alia}, the following differential consequence:
\[
\Phi_u+UX_u=\rho_{u_x}.
\]
It is obtained as follows. We subtract the result of action of the operator~$\p_x+u_x\p_u$
on the second equation of~\eqref{EqCLsOfEvolEqsSystemForReductionOf1CL}
from the partial $u_x$-derivative of the first equation of~\eqref{EqCLsOfEvolEqsSystemForReductionOf1CL}
while taking into account the contact condition~\eqref{EqContactConditionForContactTransOfGenEvolEqs}.
Moreover, the system~\eqref{EqCLsOfEvolEqsSystemForReductionOf1CL} also implies
the equation $\Phi_x+UX_x=\rho-u_x\rho_{u_x}$.
Thus, we arrive at the system
\begin{equation}\label{EqCLsOfEvolEqsSystemForReductionOf1CL2}
\Phi_x+UX_x=\rho-u_x\rho_{u_x}, \quad \Phi_u+UX_u=\rho_{u_x}, \quad \Phi_{u_x}+UX_{u_x}=0.
\end{equation}
Reversing these steps shows that the system~\eqref{EqCLsOfEvolEqsSystemForReductionOf1CL2}
implies~\eqref{EqContactConditionForContactTransOfGenEvolEqs} and~\eqref{EqCLsOfEvolEqsSystemForReductionOf1CL}.
Hence the combined system of~\eqref{EqContactConditionForContactTransOfGenEvolEqs} and~\eqref{EqCLsOfEvolEqsSystemForReductionOf1CL}
is equivalent to \eqref{EqCLsOfEvolEqsSystemForReductionOf1CL2}.

To complete the proof, it suffices to check that for any function $\rho=\rho(t,x,u,u_x)$ with $(\rho_u,\rho_{u_x})\ne(0,0)$
the system~\eqref{EqCLsOfEvolEqsSystemForReductionOf1CL2} has a solution $(X,U,\Phi)$ which satisfies
the nondegeneracy condition~\eqref{EqNondegeneracyAssumptionForContactTransOfGenEvolEqs}.

Consider first the case $\rho_{u_xu_x}\ne0$ and seek for solutions with $X_{u_x}\ne0$.
The equation $ \quad \Phi_{u_x}+UX_{u_x}=0$ implies that $\Phi_{u_x}\ne0$ and $U=-\Phi_{u_x}/X_{u_x}$.
Then the remaining equations in \eqref{EqCLsOfEvolEqsSystemForReductionOf1CL2} take the form
\begin{equation}\label{EqCLsOfEvolEqsSystemForReductionOf1CL3}
\Phi_x-\frac{X_x}{X_{u_x}}\Phi_{u_x}=\rho-u_x\rho_{u_x}, \quad \Phi_u-\frac{X_u}{X_{u_x}}\Phi_{u_x}=\rho_{u_x}.
\end{equation}
Eq.~\eqref{EqCLsOfEvolEqsSystemForReductionOf1CL3} can be considered
as an overdetermined system with respect to~$\Phi$. The compatibility condition for this system is
\[
\rho_{u_xu_x}X_x+u_x\rho_{u_xu_x}X_u+(\rho_u-u_x\rho_{uu_x}-\rho_{xu_x})X_{u_x}=0;
\]
it should be treated as an equation for~$X$.
As $\rho_{u_xu_x}\ne0$ by assumption, the equation in question has a local solution $X^0$ with $X^0_{u_x}\ne0$.
Substituting $X^0$ into~\eqref{EqCLsOfEvolEqsSystemForReductionOf1CL3} yields a compatible partial differential
system for~$\Phi$.
Take a local solution $\Phi^0$ of this system and set $U^0=-\Phi^0_{u_x}/X^0_{u_x}$.
The chosen triple $(X^0,U^0,\Phi^0)$ satisfies \eqref{EqCLsOfEvolEqsSystemForReductionOf1CL2}.

{\samepage
The nondegeneracy condition~\eqref{EqNondegeneracyAssumptionForContactTransOfGenEvolEqs} is also satisfied.
Indeed, if we assume the converse, then $U=\Psi(t,X)$ for some function~$\Psi$ of
two arguments, and \eqref{EqCLsOfEvolEqsSystemForReductionOf1CL} implies the equality
\[
\rho=\Phi_x+\Psi X_x+(\Phi_u+\Psi X_u)u_x+(\Phi_{u_x}+\Psi X_{u_x})u_{xx}=D_x(\Phi+\textstyle\int\!\Psi\,dX),
\]
i.e., $(\rho,\sigma)$ is a trivial conserved vector, which contradicts the initial assumption
on $(\rho,\sigma)$.

}

Now turn to the case when $\rho_{u_xu_x}=0$. Then
up to the equivalence of conserved vectors we can assume that $\rho_{u_x}=0$ and
$\rho_u\ne0$, where the latter condition ensures nontriviality of the associated conserved vector.
The triple $(X,U,\Phi)=(x,\rho,0)$ obviously satisfies~\eqref{EqCLsOfEvolEqsSystemForReductionOf1CL2}
and~\eqref{EqNondegeneracyAssumptionForContactTransOfGenEvolEqs}, and the result follows.
\end{proof}

\begin{corollary}\label{CorollaryOnOne0OrderConsLawOfEvolEqs}
Any pair $(\mathcal E,\mathcal L)$, where $\mathcal E$ is an equation of the form~\eqref{EqGenEvol}
and $\mathcal L$ is a nontrivial conservation law of $\mathcal E$ with the density order 0
is $G^\sim_{\rm p}$-equivalent to a pair $(\tilde{\mathcal E},\tilde{\mathcal L})$,
where $\tilde{\mathcal E}$ also is an equation of form~\eqref{EqGenEvol}
and $\tilde{\mathcal L}$ is a conservation law of $\tilde{\mathcal E}$ with the characteristic equal to $1$.
\end{corollary}

\begin{corollary}\label{CorollaryOnFormOfEvolEqsWithOne1stOrderConsLaw}
An equation~$\mathcal E$ from class~\eqref{EqGenEvol} admits a nontrivial conservation law~$\mathcal L$
with $\ord_{\rm d}\mathcal L\leqslant1$ (resp.\ $\ord_{\rm d}\mathcal L=0$)
if and only if it can be locally reduced by a contact (resp.\ point) transformation
to the form
\begin{equation}\label{EqCanonicalFormOfEvolEqsWithOne1stOrderConsLaw}
\tilde{u}_{\tilde t}=D_{\tilde x}G(\tilde t,\tilde x,\tilde{u}_0,\dots,\tilde{u}_{n-1}),\quad
G_{\tilde{u}_{n-1}}\ne0.
\end{equation}
\end{corollary}
Note that upon setting $n=3$ and $\ord_{\rm d}\mathcal L=0$ in this
corollary we recover Theorem~1.1 from~\cite{fol2}.
\begin{proof}
Fix a nontrivial conservation law~$\mathcal L$ of~$\mathcal E$ 
with $\ord_{\rm d}\mathcal L\leqslant1$ (resp.\ $\ord_{\rm d}\mathcal L=0$).
By  Theorem~\ref{TheoremOnOneConsLawOfEvolEqs} (resp.\ Corollary~\ref{CorollaryOnOne0OrderConsLawOfEvolEqs}),
the pair $(\mathcal E,\mathcal L)$ is reduced by a contact (resp.\ point) transformation to
a pair $(\tilde{\mathcal E},\tilde{\mathcal L})$,
where the equation~$\tilde{\mathcal E}$ has the form
$\tilde u_{\tilde t}=\tilde F(\tilde t,\tilde x,\tilde u_0,\dots,\tilde u_n)$
and $\tilde{\mathcal L}$ is its conservation law with the unit characteristic.
Therefore, the equality $D_{\tilde t}\tilde \rho+D_{\tilde x}\tilde \sigma=\tilde u_{\tilde t}-\tilde F$ is satisfied
for a conserved vector $(\tilde \rho,\tilde \sigma)$ from $\tilde{\mathcal L}$, i.e.,
up to a summand being a null divergence we have $\tilde \rho=\tilde u$ and $\tilde F=-D_{\tilde x}\tilde \sigma$.
To complete the proof, it suffices to put $G=-\tilde \sigma$.

Conversely, let the equation~$\mathcal E$ be locally reducible by a
contact (resp.\ point) transformation~$\mathcal T$ to the equation
$\tilde u_{\tilde t}=D_{\tilde x}G(\tilde t,\tilde x,\tilde{u}_0,\dots,\tilde{u}_{n-1})$, where $G_{\tilde{u}_{n-1}}\ne0$.
The transformed equation $\tilde u_{\tilde
t}=D_{\tilde x}G$ admits at least the conservation
law~$\tilde{\mathcal L}$ with the unit characteristic. The
preimage~$\mathcal L$ of~$\tilde{\mathcal L}$ with respect
to~$\mathcal T$ is a nontrivial conservation law of~$\mathcal E$ with
$\ord_{\rm d}\mathcal L\leqslant1$ (resp.\ $\ord_{\rm d}\mathcal
L=0$).
\end{proof}

\begin{corollary}\label{CorollaryOnFormOfEvolEqsWithTwo1thAndLowOrderConsLaws}
If an equation~$\mathcal E$ of the form~\eqref{EqGenEvol} with $n\geqslant4$ (resp.\ $n\geqslant5$)
has two linearly independent conservation laws $\mathcal L^{\rm I}$ and $\mathcal L^{\rm II}$,
where $\ord_{\rm d}\mathcal L^{\rm I}\leqslant1$ and
$\ord_{\rm d}\mathcal L^{\rm II}\leqslant n/2-1$ (resp.\ $\ord_{\rm d}\mathcal L^{\rm II}<n/2-1$)
then it can be locally reduced by a contact transformation to the form~\eqref{EqCanonicalFormOfEvolEqsWithOne1stOrderConsLaw}
where $G$ is linear fractional (resp. linear) with respect to $\tilde u_{n-1}$, i.e.,
\[
G=\frac{G_1\tilde u_{n-1}+G_0}{G_3\tilde u_{n-1}+G_2} \quad(\mbox{resp. } G=G_1\tilde u_{n-1}+G_0),
\]
where $G_0$, \dots, $G_3$ (resp.\ $G_0$ and $G_1$) are differential functions of order less than $n-1$.
If $\ord_{\rm d}\mathcal L^{\rm I}=0$ then the contact transformation in question is a prolongation of a point transformation.\looseness=-1
\end{corollary}

\begin{proof}
Without loss of generality we can assume that $\ord_{\rm d}\mathcal L^{\rm I}\leqslant\ord_{\rm d}\mathcal L^{\rm II}$.
By Theorem~\ref{TheoremOnOneConsLawOfEvolEqs} and Corollary~\ref{CorollaryOnFormOfEvolEqsWithOne1stOrderConsLaw},
the pair $(\mathcal E,\mathcal L^{\rm I})$ is reduced by a contact transformation to
a pair $(\tilde{\mathcal E},\tilde{\mathcal L}^{\rm I})$, where
the equation~$\tilde{\mathcal E}$ is of form~\eqref{EqCanonicalFormOfEvolEqsWithOne1stOrderConsLaw}
and the conservation law~$\tilde{\mathcal L}^{\rm I}$ has the density~$\tilde u$.
The transformed conservation law $\tilde{\mathcal L}^{\rm II}$ satisfies the same inequality as the original one, $\mathcal L^{\rm II}$,
i.e., $\ord_{\rm d}\tilde{\mathcal L}^{\rm II}\leqslant n/2-1$ (resp.\ $\ord_{\rm d}\tilde{\mathcal L}^{\rm II}<n/2-1$)
if $n\geqslant4$ (resp.\ $n\geqslant5$).
Below we omit tildes over the transformed variables for convenience and assume that
the conservation law~$\mathcal L^{\rm I}$ possesses the density~$u$ and, therefore,
the equation~$\mathcal E$ has the form~\eqref{EqCanonicalFormOfEvolEqsWithOne1stOrderConsLaw}.
Let $(\rho^{\rm II},\sigma^{\rm II})$ be a conserved vector associated with $\mathcal L^{\rm II}$
and $\ord \rho^{\rm II}=\ord_{\rm d}\mathcal L^{\rm II}$.
By~\eqref{cl0a}, it satisfies the condition $ \rho^{\rm II}_t+(D_xG)\delta\rho^{\rm II}/\delta u=D_x\Psi $
for some differential function~$\Psi$.
The last equality can be rewritten as
\begin{equation}\label{Grho}
\rho^{\rm II}_t-GD_x(\delta\rho^{\rm II}/\delta u)=D_x\Phi,
\end{equation}
where $\Phi=\Psi-G \delta\rho^{\rm II}/\delta u$.
Note that $D_x(\delta\rho^{\rm II}/\delta u)\ne0$ because
otherwise the conservation laws $\mathcal L^{\rm I}$ and $\mathcal L^{\rm II}$ are linearly dependent.
As $\ord \rho^{\rm II}_t<n-1$, $\ord G=n-1$ and $\ord D_x(\delta\rho^{\rm II}/\delta u)\leqslant n-1$
(resp.\ $\ord D_x(\delta\rho^{\rm II}/\delta u)<n-1$), we have $\ord\Phi<n-1$.
Finally, as $D_x (\delta\rho^{\rm II}/\delta u)$ must be linear in the highest-order $x$-derivative of $u$ it contains,
expressing $G$ from (\ref{Grho}) and taking into account the above inequalities for $\ord D_x(\delta\rho^{\rm II}/\delta u)$
immediately yields the desired result.
\end{proof}

\begin{theorem}\label{TheoremOnTwoConsLawOfEvolEqs}
Let $\mathcal E$ be an equation of the form~\eqref{EqGenEvol} and
$\mathcal L^{\rm I}$ and $\mathcal L^{\rm II}$ be linearly
independent conservation laws of~$\mathcal E$ of density order~0.
Any such triple $(\mathcal E,\mathcal L^{\rm I},\mathcal L^{\rm II})$
is $\smash{G^\sim_{\rm p}}$-equivalent to a triple
$(\tilde{\mathcal E},\tilde{\mathcal L^{\rm I}},\tilde{\mathcal L^{\rm II}})$, where $\tilde{\mathcal E}$ is an equation from
the same class \eqref{EqGenEvol} that admits conservation laws
$\tilde{\mathcal L^{\rm I}}$ and $\tilde{\mathcal L^{\rm II}}$ with the
characteristics~equal to 1 and $\tilde x$, respectively.
\end{theorem}

\begin{proof}
Let $(\rho^i,\sigma^i)\in\mathcal L^i$ and $\ord \rho^i=0$, $i={\rm I,II}$.
Then $\gamma^i=\rho^i_u$ is the characteristic of $\mathcal L^i$, $i={\rm I,II}$.
Moreover, $\gamma^{\rm I}$ and $\gamma^{\rm II}$ are linearly independent differential functions
in view of the linear independence of conservation laws $\mathcal L^{\rm I}$ and $\mathcal L^{\rm II}$.
Therefore, we have $(\lambda_x,\lambda_u)\ne(0,0)$, where $\lambda=\gamma^{\rm II}/\gamma^{\rm I}$.
(Indeed, otherwise the substitution of these characteristics into \eqref{cosym}
would imply that $\lambda_t=0$ as well, i.e., the characteristics $\gamma^{\rm I}$ and $\gamma^{\rm II}$ would be linearly dependent.)

We will prove the existence of (and, in fact, construct) a point equivalence
transformation of the form~\eqref{EqPointTransOfGenEvolEqs} with
$T(t)=t$ such that the transformed conserved vectors
$(\tilde \rho^{\rm I},\tilde \sigma^{\rm I})$ and $(\tilde \rho^{\rm II},\tilde \sigma^{\rm II})$
are equivalent to the conserved vectors with the densities $\tilde u$ and $\tilde x\tilde u$, respectively.
In other words, we want to have
$\tilde \rho^{\rm I}=\tilde u+D_{\tilde x}\Phi$ and
$\tilde \rho^{\rm II}=\tilde x\tilde u+D_{\tilde x}\Psi$ for some functions $\Phi=\Phi(t,x,u)$ and $\Psi=\Psi(t,x,u)$.
In the old coordinates these conditions
take the form
$D_x\Phi+UD_xX=\rho^{\rm I}$ and $D_x\Psi+XUD_xX=\rho^{\rm II}$.
Splitting them with respect to~$u_x$ yields
\[
\begin{array}{l}\Phi_x+UX_x=\rho^{\rm I},\\[1ex] \Phi_u+UX_u=0\end{array}
\quad\mbox{and}\quad
\begin{array}{l}\Psi_x+XUX_x=\rho^{\rm II},\\[1ex]\Psi_u+XUX_u=0.\end{array}
\]
After the elimination of $\Phi$ and $\Psi$ from these systems through
cross-differentiation, we arrive at the conditions
$X_xU_u-X_uU_x=\rho^{\rm I}_u$ and $\rho^{\rm I}_uX=\rho^{\rm II}_u$.
If we set $X=\lambda=\rho^{\rm II}_u/\rho^{\rm I}_u$ then $(X_x,X_u)\ne(0,0)$.
This ensures existence of a function $U=U(t,x,u)$ which locally satisfies the
equation $X_xU_u-X_uU_x=\rho^{\rm I}_u$.
It is obvious that the so chosen functions~$X$ and~$U$ are functionally independent and that
the above systems are then compatible with respect to $\Phi$ and $\Psi$,
and hence the point transformation we sought for does exist.
\end{proof}

\begin{corollary}\label{CorollaryOnFormOfEvolEqsWithTwo0thOrderConsLaws}
An equation~$\mathcal E$ of the form~\eqref{EqGenEvol} has (at least) two
linearly independent conservation laws of density order 0 if and
only if it can be locally reduced by a point transformation to the
form
\begin{equation}\label{cf1x}
\tilde{u}_{\tilde t}=D_{\tilde x}^2H(\tilde t,\tilde x,\tilde{u}_0,\dots,\tilde{u}_{n-2}), \quad
H_{\tilde{u}_{n-2}}\ne0.
\end{equation}
\end{corollary}

\begin{proof}
If $\mathcal E$ of the form~\eqref{EqGenEvol} admits (at least) two
linearly independent conservation laws of density order~0, then by Theorem~\ref{TheoremOnTwoConsLawOfEvolEqs}
we can assume (modulo a suitably chosen point transformation, if necessary)
that $\mathcal E$ has the conservation laws~$\mathcal L^{\rm I}$ and~$\mathcal L^{\rm II}$
with the characteristics~1 and~$x$, respectively.
Then there exist conserved vectors
$(\rho^{\rm I},\sigma^{\rm I})\in\mathcal L^{\rm I}$ and
$(\rho^{\rm II},\sigma^{\rm II})\in\mathcal L^{\rm II}$ such that
\[
D_t\rho^{\rm I}+D_x\sigma^{\rm I}=u_t-F, \quad D_t\rho^{\rm II}+D_x\sigma^{\rm II}=x(u_t-F).
\]
Up to the equivalence of conserved vectors modulo trivial ones we have $\rho^{\rm I}=u$ and $\rho^{\rm II}=xu$.
Hence $D_x\sigma^{\rm I}=-F$ and $D_x\sigma^{\rm II}=-xF$.
Combining these equalities, we find that $\sigma^{\rm I}=-D_x(\sigma^{\rm II}-x\sigma^{\rm I})$, i.e.,
$F=D_x^2(\sigma^{\rm II}-x\sigma^{\rm I})$.
As a result, we can represent the equation~$\mathcal E$ in the form $u_t=D_x^2H$,
where $H=\sigma^{\rm II}-x\sigma^{\rm I}$, $\ord H=n-2$.

Conversely, assume that~$\mathcal E$ is reduced to the equation
$\tilde u_{\tilde t}=D_{\tilde x}^2H(\tilde t,\tilde x,\tilde
u,\dots,\tilde u_{n-2})$, where $H_{\tilde u_{n-2}}\ne0$,
through a point transformation~$\mathcal T$. The transformed
equation $\tilde u_{\tilde t}=D_{\tilde x}^2H$ admits at least
two linearly independent conservation laws, in particular,
those with the characteristics~1 and~$\tilde x$. Their preimages under~$\mathcal T$ are linearly
independent conservation laws of~$\mathcal E$ whose density orders
are zero, and the result follows.
\end{proof}

\begin{corollary}\label{CorollaryOnFormOfEvolEqsWithTwo1thAndLowOrderConsLaws2}
If an equation~$\mathcal E$ of the form~\eqref{EqGenEvol} with
$n\geqslant5$ (resp.\ $2\leqslant n\leqslant4$) has two linearly independent conservation laws
$\mathcal L^{\rm I}$ and $\mathcal L^{\rm II}$ with 
$\ord_{\rm d}\mathcal L^{\rm I}\leqslant1$ and $\ord_{\rm d}\mathcal L^{\rm II}<n/2-1$
(resp.\ $\ord_{\rm d}\mathcal L^{\rm I}=\ord_{\rm d}\mathcal L^{\rm II}=0$),
then the right-hand side~$F$ of $\mathcal E$ has the form
\[
F=F_3 u_n+F_2 u_{n-1}^2+F_1u_{n-1}+F_0,
\]
where $F_0$, \dots, $F_3$ are differential functions of order less than $n-1$.
\end{corollary}

\begin{remark}
If in the proof of Corollary~\ref{CorollaryOnFormOfEvolEqsWithTwo0thOrderConsLaws}
we replace the conservation laws~$\mathcal L^{\rm I}$ and~$\mathcal L^{\rm II}$ by linear combinations thereof,
$\hat{\mathcal L}^{\rm I}=a_{11}\mathcal L^{\rm I}+a_{12}\mathcal L^{\rm II}$ and
$\hat{\mathcal L}^{\rm II}=a_{21}\mathcal L^{\rm I}+a_{22}\mathcal L^{\rm II}$,
where $a_{ij}$, $i,j=1,2$, are arbitrary constants such that $a_{11}a_{22}-a_{12}a_{21}\ne0$,
then the associated equations of the form~\eqref{cf1x} are related through the transformation
\[
\hat t =\tilde t, \quad
\hat x=\frac{a_{22}\tilde x+a_{21}}{a_{12}\tilde x+a_{11}}, \quad
\hat u=\frac{(a_{12}\tilde x+a_{11})^3}{a_{11}a_{22}-a_{12}a_{21}}\tilde u, \quad
\hat H=\frac{a_{11}a_{22}-a_{12}a_{21}}{a_{12}\tilde x+a_{11}}H,
\]
where $\hat t$, $\hat x$, $\hat u$ and the differential function~$\hat H$ correspond
to the conservation laws $\hat{\mathcal L}^{\rm I, II}$.
Such transformations, considered for all admissible values of $a_{ij}$, $i,j=1,2$,
form a subgroup~$\mathcal G$ of the point equivalence group
for the class of equations of the form~\eqref{cf1x}.
Thus, up to the $\mathcal G$-equivalence we can assume that the form~\eqref{cf1x} of the equation~$\mathcal E$
is associated with the two-dimensional {\em subspace} spanned by its conservation laws $\mathcal L^{\rm I}$
and $\mathcal L^{\rm II}$ rather than with $\mathcal L^{\rm I}$
and $\mathcal L^{\rm II}$ {\em per se}.
\end{remark}

\section{Examples: third-order evolution equations}%
\label{SectionOnExamplesOfThird-orderEvolutionEqs}

\begin{example}
We start with the so-called Harry Dym (HD) equation, see e.g.\ \cite[Section~20.2]{Ibragimov1985}
and references therein for more details:
\[
u_t=u^3 u_{xxx}.
\]
The subspace of its conservation laws of density order not greater than one
is five-dimensional and generated by the zero-order conservation laws $\mathcal L^i$, $i={\rm I},\dots, {\rm IV}$, with the densities
$\rho^{\rm I}=u^{-2}$, $\rho^{\rm II}=xu^{-2}$, $\rho^{\rm III}=x^2u^{-2}$, and $\rho^{\rm IV}=u^{-1}$,
and the first-order conservation law $\mathcal L^V$ with the density $\rho^{\rm V}=u_x^2u^{-1}$.

The first three densities agree in the sense that $\rho^{\rm II}/\rho^{\rm I}=\rho^{\rm III}/\rho^{\rm II}$.
Hence upon introducing new variables $\tilde t=-2t\sign u$, $\tilde x=x$, $\tilde u=u^{-2}$
obtained by applying Theorem~\ref{TheoremOnTwoConsLawOfEvolEqs} to~$\mathcal L^{\rm I}$ and~$\mathcal L^{\rm II}$,
the HD equation can be rewritten in an even more specific than \eqref{cf1x}, and also well-known, conservative form
$
\tilde u_{\tilde t}=D_{\tilde x}^3(\tilde u^{-1/2}).
$
(We transformed~$t$ above in order to simplify the transformed equation.)
The transformed equation obviously admits conservation laws with the characteristics
equal to 
1, $x$ and $x^2$.

For the pair of conservation laws~$\mathcal L^{\rm IV}$ and~$\mathcal L^{\rm I}$ Theorem~\ref{TheoremOnTwoConsLawOfEvolEqs} yields the transformation
$\tilde t=t$, $\tilde x=-2/u$, $\tilde u=x/2$ which maps the HD equation into the equation
\[
\tilde u_{\tilde t}=D_{\tilde x}^2\left(\frac1{2\tilde x^3\tilde u_{\tilde x}^2}\right).
\]

The conservation law $\mathcal L^{\rm V}$ is mapped
into a conservation law with the characteristic $1$
by the contact transformation
$\tilde t=t$,
$\tilde x=u_x^2/u$,
$\tilde u=u-2u/u_x$,
$\tilde u=u^2/u_x^3$
constructed using the method from the proof of Theorem~\ref{TheoremOnOneConsLawOfEvolEqs}.
The corresponding transformed equation reads
\[
\tilde u_{\tilde t}=D_{\tilde x}\left(\frac{-\tilde x^8\tilde u_{\tilde x}^6}
{4(2\tilde x\tilde u_{\tilde x\tilde x}+3\tilde u_{\tilde x})^2}\right).
\]
\end{example}

\begin{example}
Consider now the class of KdV-type equations
\begin{equation}\label{kdvtype}
u_t=u_{xxx}+f(u)u_x.
\end{equation}
Any equation from this class admits at least three conservation laws $\mathcal L^i$, $i={\rm I},\dots, {\rm III}$,
with the densities $\rho^{\rm I}=u$, $\rho^{\rm II}=u^2/2$, $\rho^{\rm III}=-u_x^2/2+\check f(u)$,
where  $\p\hat f/\p u=f$, $\p\check f/\p u=\hat f$.
It is straightforward to verify
that if $\p^3 f/\p u^3\neq 0$ these conservation laws form a basis in the space of
the conservation laws of density order not greater than one.

{\samepage
The reduction \eqref{kdvtype} to the
form~\eqref{EqCanonicalFormOfEvolEqsWithOne1stOrderConsLaw} using
$\mathcal L^{\rm I}$ (resp.\ $\mathcal L^{\rm II}$) according to
Theorem~\ref{TheoremOnOneConsLawOfEvolEqs} is immediate. The conservation law
$\mathcal L^{\rm I}$ gives rise to the identity transformation and the
representation $u_t=D_x(u_{xx}+\hat f(u))$ for \eqref{kdvtype}. 
The transformation associated with $\mathcal L^{\rm II}$ is
$\tilde t=t$, $\tilde x=x$ and $\tilde u=\rho^{\rm II}=u^2/2$.
It maps equation~\eqref{kdvtype} into
\[
\tilde u_{\tilde t}=D_{\tilde x}\left(\tilde u_{\tilde x\tilde x}
-\frac34\frac{\tilde u_{\tilde x}^2}{\tilde u}
+\varepsilon\sqrt{2\tilde u}\hat f(\varepsilon\sqrt{2\tilde u})-\check f(\varepsilon\sqrt{2\tilde u}) \right),
\]
where $\varepsilon=\sign u$.

}

Now consider the conservation laws $\mathcal L^{\rm I}$ and $\mathcal L^{\rm II}$ and apply
Theorem~\ref{TheoremOnTwoConsLawOfEvolEqs}.
We can directly follow the procedure from the proof
of this theorem and set $\tilde t=t$, $\tilde x=\rho^{\rm II}_u=u$ and $\tilde u=x$.
This is nothing but the hodograph transformation interchanging~$x$
and~$u$. It reduces equation~\eqref{kdvtype} to the equation (cf.\ \cite{fol2})
\[
\tilde u_{\tilde t}=D_{\tilde x}^2\left(\frac1{2\tilde u_{\tilde x}^2}-\check f(\tilde x) \right).
\]
\end{example}

\begin{example}
The KdV equation, i.e., equation~\eqref{kdvtype} with
$f(u)=u$, possesses one more linearly independent zero-order conservation law $\mathcal L^{\rm IV}$ with
the density $\rho^{\rm IV}=xu+tu^2/2$, cf.\ \cite{Miura&Gardner&Kruskal1968}. This
gives more possibilities for reduction to the
forms~\eqref{EqCanonicalFormOfEvolEqsWithOne1stOrderConsLaw}
and~\eqref{cf1x}.

In analogy with the previous example, we find that the transformation associated with $\mathcal L^{\rm IV}$ is
$\tilde t=t$, $\tilde x=x$ and $\tilde u=\rho^{\rm IV}=xu+tu^2/2$.
It maps the KdV equation into
\[
\tilde u_{\tilde t}=D_{\tilde x}\left(\tilde u_{\tilde x\tilde x}
-\frac32\frac{\tilde t}Z\tilde u_{\tilde x}^2-3\frac{\tilde x}{Z}\tilde u_{\tilde x}
\pm\frac{Z^{3/2}}{3\tilde t^2}+3\frac{\tilde u}Z-\frac{\tilde x}{\tilde t}\tilde u-\frac{\tilde x^3}{3\tilde t^2}
\right),
\]
where $Z=\tilde x^2+2\tilde t\tilde u$.

For the pair of the conservation laws $\mathcal L^{\rm I}$ and~$\mathcal L^{\rm IV}$ we have the transformation of the form
$\tilde t=t$, $\tilde x=\rho^{\rm IV}_u=x+tu$ and $\tilde u=u$, and
the transformed equation reads
\[
\tilde u_{\tilde t}=D_{\tilde x}\left(\frac{\tilde u_{\tilde x\tilde x}}{(1-t\tilde u_{\tilde x})^3} \right)
=D_{\tilde x}^2\left(\frac{(1-t\tilde u_{\tilde x})^{-2}}{2t}\right).
\]

Another pair of the conservation laws, $\mathcal L^{\rm II}$ and~$\mathcal L^{\rm IV}$,
gives rise to a more complicated transformation
$\tilde t=t$, $\tilde x=x/u+t$, $\tilde u=u^3/3$, and
a more cumbersome transformed equation,
\[
\tilde u_{\tilde t}=D_{\tilde x}^2\left(\frac
{((\tilde x-\tilde t)\tilde u_{\tilde x}+6\tilde u)\tilde u_{\tilde x}}
{2((\tilde x-\tilde t)\tilde u_{\tilde x}+3\tilde u)^2}
\right).
\]
Note that exhaustive lists of one- and two-dimensional subspaces of zero-order conservation laws of the KdV equation
that are not equivalent with respect to the Lie point symmetry group of the latter
are $\{\langle\mathcal L^{\rm I}\rangle,\langle\mathcal L^{\rm II}\rangle,\langle\mathcal L^{\rm IV}\rangle\}$
and $\{\langle\mathcal L^{\rm I},\mathcal L^{\rm II}\rangle,\langle\mathcal L^{\rm I},
\mathcal L^{\rm IV}\rangle,\langle\mathcal L^{\rm II},\mathcal L^{\rm IV}\rangle\}$,
respectively.
Therefore, the above description of normal forms~\eqref{EqCanonicalFormOfEvolEqsWithOne1stOrderConsLaw}
and~\eqref{cf1x} related to zero-order conservation laws 
of the KdV equation is complete modulo the action of the Lie point symmetry group
of the KdV equation, cf.\ Remark~1.
\end{example}

\begin{example}
The Schwarzian KdV equation
\[
u_t=u_{xxx}-\frac32\frac{u_{xx}^2}{u_x}
\]
possesses no zero-order conservation laws.
The subspace of its first-order conservation laws
is spanned by the conservation laws $\mathcal L^i$, $i={\rm I},\dots, {\rm III}$, with
the densities $\rho^{\rm I}=1/u_x$, $\rho^{\rm II}=u/u_x$, $\rho^{\rm III}=u^2/u_x$.
For transforming $\mathcal L^{\rm I}$ into a conservation law with the density $u$, we construct, following the proof
of Theorem~\ref{TheoremOnOneConsLawOfEvolEqs}, the contact transformation
\[
\tilde t=t,\quad
\tilde x=u_x,\quad
\tilde x=\frac{2x}{u_x^2}-\frac{2u}{u_x^3},\quad
\tilde u=-\frac{4x}{u_x^3}+\frac{6u}{u_x^4},
\]
which maps the Schwarzian KdV equation into
\[
\tilde u_{\tilde t}=D_{\tilde x}\left(\frac{-4\tilde x^{-5}}
{(\tilde x^2\tilde u_{\tilde x\tilde x}+6\tilde x\tilde u_{\tilde x}+6\tilde u)^2} \right).
\]
\end{example}

Further examples of normal forms for low-order nonlinear evolution equations,
including physically relevant examples like the nonlinear 
diffusion-convection equations, can be found in~\cite{Popovych&Samoilenko2008}.

\section{Conservation laws of linear evolution equations}%
\label{SectionOnCLsOfLinEvolEqs}

Any linear partial differential equation admits conservation laws
whose characteristics depend on independent variables only and run through the set of solutions of
the adjoint equation.
The corresponding conserved vectors are linear with respect to the unknown function and its derivatives.
It is natural to call the conservation laws of this kind \emph{linear} \cite[Section~5.3]{Olver1993}.
Let~us stress that, following the literature, here and below we allow for a slight abuse of terminology by calling a conservation law
linear (resp.\ quadratic) when it contains a conserved vector which is linear (resp.\ quadratic) in the totality of variables $u_0,u_1,u_2,\dots$.

The problem of describing other kinds of conservation laws for general linear partial differential equations
is quite difficult. However,
it can be solved for certain special classes of equations including linear $(1+1)$-dimensional evolution equations.

Consider an equation $\mathcal E$ of form~\eqref{EqGenEvol}, where the function $F$ is linear in $u_0$, \dots, $u_n$,
i.e.,
\[
F=\mathfrak F\,u=\sum_{i=0}^n A^i(t,x)u_i, \quad\mbox{where}\quad
\mathfrak F=\sum_{i=0}^n A^i(t,x)D_x^i, \quad A^n\ne0.
\]
Thus, the equation~$\mathcal E$  reads
\begin{equation}\label{lineq}
u_t=\mathfrak F\,u.
\end{equation}
Then the condition~\eqref{cosym} for cosymmetries takes the form
\[
D_t\gamma+\mathfrak F^\dagger\gamma=0\bmod\check{\mathcal E}, \quad\mbox{where}\quad
\mathfrak F^\dagger=\sum_{i=0}^n (-D_x)^i\circ A^i(t,x).
\]
The operator~$\mathfrak F^\dagger$ is the formal adjoint of~$\mathfrak F$.
Writing out the condition~\eqref{cosym} yields
\begin{equation}\label{EqForCharsOfLinEvolEqs}
\gamma_t +\sum_k\gamma_{u_k}\sum_{i=0}^n\sum_{j=0}^k\binom kj
A^i_{k-j}u_{i+j}+\sum_{i=0}^n(-1)^i\sum_{s=0}^i\binom is
A^i_{i-s}D_x^s\gamma =0,
\end{equation}
where $A_j^i=\p^j A^i/\p x^j$.

A function $v=v(t,x)$ is a cosymmetry of the equation $\mathcal E$
if and only if it is a solution of the {\em adjoint} equation
$\mathcal E^*$:
\begin{equation}\label{ale}
v_t+\mathfrak F^\dagger v=0.
\end{equation}
Any cosymmetry of~$\mathcal E$ that does not depend on $u$ and the derivatives thereof
is a characteristic of a linear conservation law of~$\mathcal E$,
and any linear conservation law of~$\mathcal E$ has a characteristic of this form.
Namely, a solution $v=v(t,x)$ of the adjoint equation~$\mathcal E^*$ corresponds to
the conserved vector $(\rho,\sigma)$ of~$\mathcal E$
with $\rho=v(t,x) u$ and $\sigma=\sum_{i=0}^{n-1}\sigma^i(t,x)u_i$.
The coefficients $\sigma^i$ are found recursively from the
equations
\begin{equation}\label{sigma-lin}
\sigma^{n-1}=-v A^n, \quad \sigma^i=-v A^{i+1}-\sigma^{i+1}_x, \quad
i=n-2,\dots,0.
\end{equation}

It turns out that {\em all} cosymmetries of {\em even-order} equations (\ref{lineq}) are of this form.

\begin{theorem}\label{lin-cl-th1}
For any linear $(1+1)$-dimensional evolution equation of even order,
all its cosymmetries depend only on $x$ and $t$, and the space of all cosymmetries
is isomorphic to the solution space of the associated adjoint equation.
\end{theorem}

\begin{proof}
Suppose that there exists a $\nu\in\mathbb N\cup\{0\}$ such that
$\gamma_{u_\nu}\ne0$ and denote \[r=\max\{\nu\in\mathbb
N\cup\{0\}\mid \gamma_{u_\nu}\ne0\}.\]

For even $n$ vanishing of the coefficient at $u_{n+r}$
in~\eqref{EqForCharsOfLinEvolEqs} yields the equation $2A^n
\gamma_{u_r}=0$, whence $\gamma_{u_r}=0$. This contradicts the original
assumption, and hence $\gamma$ depends only on~$t$ and~$x$.
\end{proof}

\begin{corollary}\label{peolin}
For any linear $(1+1)$-dimensional evolution equation of even order
its space of conservation laws is exhausted by linear ones
and is isomorphic to the solution space of the corresponding adjoint equation.
\end{corollary}

For odd $n$ things become somewhat more involved.

\begin{theorem}\label{lin-cl-th2}
For any linear $(1+1)$-dimensional evolution equation of odd order,
all its cosymmetries are affine in the totality of variables $u_0,u_1,u_2,\dots$.
\end{theorem}

\begin{proof}
In contrast with the case of even $n$, now the coefficient at $u_{n+r}$
in~\eqref{EqForCharsOfLinEvolEqs} vanishes identically. Requiring the coefficient at $u_{n+r-1}$
in~\eqref{EqForCharsOfLinEvolEqs} to vanish yields
\[
n A^n_0 D_x\gamma_{u_r}=(2A^{n-1}_0+(r-1-n)A^n_1)\gamma_{u_r},
\]
so $\gamma_{u_r}$ depends only on~$t$ and~$x$.
Using this result while evaluating the coefficient at $u_{n+r-2}$ yields
\[
n A^n_0
D_x\gamma_{u_{r-1}}=(2A^{n-1}_0+(r-1-n)A^n_1)\gamma_{u_{r-1}}+\psi^{r-1},
\]
where $\psi^{r-1}$ is a function of~$t$ and~$x$, which is expressed via $\gamma_{u_r}$ and $A^i$;
the explicit form  of~$\psi^{r-1}$ is not important here.
Thus, $\gamma_{u_{r-1}}$ also depends only on~$t$ and~$x$.
Iterating the above procedure allows us to conclude that
the function~$\gamma$ is affine in $u_0,\dots,u_r$, that is,
\begin{equation}\label{gammarep}
\gamma=\Gamma u+v(t,x), \quad \Gamma=\sum_{k=0}^r g^k(t,x)D_x^k,
\end{equation}
and the result follows.
\end{proof}

Note that if we restrict ourselves to the case of polynomial cosymmetries for
equations with constant coefficients, then upon combining Theorems~\ref{lin-cl-th1} and~\ref{lin-cl-th2}
we recover (part of) Proposition~1 of~\cite{gal}.

For the equations (\ref{lineq}) with $\mathfrak F^\dagger=-\mathfrak F$
the determining equations for cosymmetries and for characteristics of generalized
symmetries coincide. This observation in conjunction with~Theorem~\ref{lin-cl-th2}
implies the following assertion (cf.~\cite{shsh}).

\begin{corollary}\label{asa}
For any linear $(1+1)$-dimensional evolution equation (\ref{lineq})
of odd order such that $\mathfrak F^\dagger=-\mathfrak F$, all its
generalized symmetries are affine in $u_j$ for all $j$.
\end{corollary}

Now let us get back to the general case of \eqref{lineq} with odd $n$. 
Substituting the representation \eqref{gammarep} for $\gamma$
into~\eqref{EqForCharsOfLinEvolEqs} reveals that 
$v=v(t,x)$ satisfies the adjoint equation~(\ref{ale}) which is decoupled from the equations 
for $g^i$. Thus, $v$ is a cosymmetry {\em per se}.
Just as before, to any such cosymmetry 
there corresponds a linear conservation law with the density
$\rho=v(t,x) u$. However, the issue of existence of conservation laws associated
with cosymmetries linear in $u_j$ is nontrivial. 

Indeed, let $\gamma=\Gamma u$. As we want $\gamma$ to be a characteristic of a conservation law,
we should require that $\gamma\in\mathop{\rm Im}\delta/\delta u$
(cf.\ Section~\ref{SectionOnAuxiliaryStatements}).
Hence, the operator~$\Gamma$ should be formally self-adjoint and, in particular, its order should be even
(note, however, that if $\Gamma$ is not formally self-adjoint, we can take its formally self-adjoint part
$\tilde\Gamma=(\Gamma+\Gamma^\dagger)/2$; $\tilde\gamma=\tilde\Gamma u$ is easily verified to be a cosymmetry if so is $\gamma$).
The density of the conservation law associated with the characteristic~$\gamma$
reads, up to the usual addition of a total $x$-derivative of something,
$\rho=\frac12u\Gamma u$.
Without loss of generality we can also assume the corresponding flux
to be quadratic in $u_0$, $u_1$, \dots, see Theorem~5.104 of~\cite{Olver1993}, so the
conservation law in question is \emph{quadratic}, and we obtain the following result.

\begin{theorem}\label{poolin}
For any linear $(1+1)$-dimensional evolution equation of odd order,
the space of its conservation laws is spanned by linear and quadratic ones.
\end{theorem}

For linear conservation laws with the densities of the form $\rho=v(t,x)u$ where $v$ solves the adjoint equation
we still have (\ref{sigma-lin}).

Now turn to the quadratic conservation laws.
The differential function $\Gamma u$ is a characteristic of a conservation law for~$\mathcal E$ if and only if
the operator~$\Gamma$ satisfies the following equivalent conditions:

1) it maps the solutions of the equation~$\mathcal E$ into solutions of the adjoint equation~$\mathcal E^*$;

2) $\p\Gamma/\p t+\Gamma\mathfrak F+\mathfrak F^\dagger\Gamma=0$;

\nopagebreak

3) $(\p_t+\mathfrak F^\dagger)\Gamma=\Gamma(\p_t-\mathfrak F)$,
i.e., the operator $\Gamma(\p_t-\mathfrak F)$ is formally skew-adjoint.

Note that if the operator~$\mathfrak F$ is formally skew-adjoint ($\mathfrak
F^\dagger=-\mathfrak F$) then the operators $\p_t-\mathfrak F$ and
$\Gamma$ commute: $[\p_t-\mathfrak F,\Gamma]=0$, i.e., $\Gamma$ is a
symmetry operator for the equation~$\mathcal E$.

Any linear equation admits a symmetry $u\partial_u$, and the
associated operator $\Gamma$ is the identity operator which is
obviously formally self-adjoint. Combining this result with the
above we obtain the following assertion.

\begin{proposition}
Any linear $(1+1)$-dimensional evolution equation (\ref{lineq}) of
odd order such that $\mathfrak F^\dagger=-\mathfrak F$ possesses a
conservation law with the density $\rho=u^2$.
\end{proposition}

Moreover, linear $(1+1)$-dimensional evolution equations of odd order can possess
infinite series of quadratic conservation laws of arbitrarily high orders, as illustrated by the following example.

\begin{example}
Consider the equation
\begin{equation}\label{e3}
u_t=u_{xxx}.
\end{equation}
It is straightforward to verify that in this case the determining equations
for cosymmetries and characteristics of (generalized) symmetries coincide
because Eq.~\eqref{e3} is identical with its adjoint.

Denote by $\mathcal{S}$ the space of all generalized symmetries of~\eqref{e3}
and let $\mathcal{Q}$ be the space of symmetries of the form $f(t,x)\p_u$, where $f$
solves \eqref{e3}: $f_t=f_{xxx}$.
By Corollary~\ref{asa}
the quotient space $\mathcal{S}/\mathcal{Q}$ is exhausted by linear generalized symmetries.
Successively solving the determining equations (cf.\ e.g.\ \cite{agn}) we find that
the space $\mathcal{S}/\mathcal{Q}$ is spanned by the symmetries of the form
$(D_x^k \Upsilon^l\,u)\p_u$, where $k,l=0,1,2,\dots$ and $\Upsilon=x + 3 t D_x^2$.

As the determining equations for symmetries and cosymmetries of~\eqref{e3} coincide,
the space of cosymmetries for \eqref{e3} is spanned by the following objects:

1) the cosymmetries of the form $f(t,x)$ where $u=f(t,x)$ is any solution of~\eqref{e3};

2) the cosymmetries of the form $D_x^k \Upsilon^l\,u$, where $k,l=0,1,2,\dots$.

Any cosymmetry of the first kind is associated with a conservation
law with the density $\rho=f(t,x) u$. As for cosymmetries of the
second kind, only those with even $k=2m$ are
characteristics of the conservation laws. The conservation laws in
question can (modulo trivial ones) be chosen to be quadratic,
with the densities $\rho_{lm}=\frac12 u\,D_x^{m} \Upsilon^l
D_x^{m}u$ and the density orders $l+m$, $l,m=0,1,2\dots$.
\end{example}

However, there also exist linear $(1+1)$-dimensional evolution equations
of odd order which have no quadratic conservation laws.

\begin{example}
The operator $\mathfrak F=D_x^3+x$ associated with the equations
\begin{equation}\label{e4}
u_t=u_{xxx}+xu
\end{equation}
is not formally skew-adjoint. Equation~\eqref{e4} possesses
nontrivial symmetries which are linear combinations of the operators
$((D_x^3+x)^k(D_x+t)^l\,u)\p_u$ but they cannot be employed for
construction of quadratic conservation laws of~\eqref{e4} in the
above fashion.

In fact, all cosymmetries of~\eqref{e4} depend only on $x$ and $t$, and
therefore this equation has no quadratic conservation laws.

Indeed,
using the proof by contradiction,
suppose that \eqref{e4} has a
cosymmetry $\gamma=\Gamma u$, and $\ord\gamma=r$, i.e., $g^r\ne0$.
The condition $\p\Gamma/\p t+\Gamma\mathfrak F+\mathfrak
F^\dagger\Gamma=0$ implies the following system of determining
equations for the coefficients of~$\Gamma$:
\begin{gather}\label{e4DetEq1}
3g^i_x=(i+3)g^{i+3}+g^{i+2}_t-g^{i+2}_{xxx}+2xg^{i+2}-3g^{i+1}_{xx},\quad i=1,\dots,r,
\\ \label{e4DetEq2}
2g^2+g^1_t-g^1_{xxx}+2xg^1-3g^0_{xx}=0,
\\ \label{e4DetEq3}
g^1+g^0_t-g^0_{xxx}+2xg^0=0,
\end{gather}
where the functions $g^{r+3}$, $g^{r+2}$ and $g^{r+1}$ vanish by
definition. We successively integrate \eqref{e4DetEq1} starting from
the equations with the greatest value of $i$ and going down.
The equations for $i=r$ and $i=r-1$ imply that the
coefficients $g^r$ and $g^{r-1}$ depend on $t$ but not on $x$.
Proceeding by induction, we find that for any $j=0,\dots,r$ the
function $g^{r-j}$ is a polynomial in~$x$ of
degree~$2[j/2]$. The ratio of the coefficient at the highest power of $x$ in
$g^{r-j}$ to $g^r$ (resp.\ $g^{r-1}$) is a constant if $j$
is even (resp.\ odd). Then \eqref{e4DetEq2} and~\eqref{e4DetEq3}
imply $g^r=0$ and $g^{r-1}=0$. This contradicts our assumption that
$g^r\ne0$, and the result follows.
\end{example}

\section{Conclusions}\label{SectionConclusion}

In this paper we have presented normal forms for the evolution equations
in two independent variables possessing low-order conservation laws, see
Theorems~\ref{TheoremOnOneConsLawOfEvolEqs} and \ref{TheoremOnTwoConsLawOfEvolEqs},
and Corollaries
\ref{CorollaryOnFormOfEvolEqsWithOne1stOrderConsLaw}--
\ref{CorollaryOnFormOfEvolEqsWithTwo1thAndLowOrderConsLaws2}
for details. Using these normal forms considerably
simplifies the construction of nonlocal variables associated
with the conservation laws in question and hence the study of the Abelian
coverings and nonlocal symmetries, including potential symmetries,
for the equations in question in spirit of
\cite{Bocharov_et_al_1997, Popovych&Ivanova2004CLsOfNDCEs,
Popovych&Kunzinger&Ivanova2008,Kunzinger&Popovych2008,Sergyeyev2000},
and references therein. As these
normal forms are associated, up to a certain natural equivalence (see Remark 1),
with the {\em subspaces} spanned by conservation laws rather
than conservation laws {\em per se}, we are naturally led to pose the problem
of classification of inequivalent subspaces of (low-order) conservation laws
for the classes or special cases of evolution equations of interest.

As for the {\em linear} evolution equations in two independent variables,
we have shown that their conservation laws are (modulo trivial conserved vectors, of course)
at most quadratic in the dependent variable and the derivatives thereof, see Theorem~\ref{lin-cl-th1}.
Moreover,  for the linear evolution equations of even order their conservation laws are at most linear
in these quantities, and the associated densities can be chosen to have the form of a product
of the dependent variable with a solution of the adjoint equation (Theorem~\ref{lin-cl-th2}).
It is natural to ask whether similar results can be obtained for more general linear PDEs (cf.\ \cite{shsh}
for the case of symmetries),
and we intend to address this issue in our future work.

\subsection*{Acknowledgements}

The research of R.O.P.\ was supported by the project P20632 of the Austrian Science Fund.
The research of A.S. was supported in part by the Ministry of Education, Youth and Sports of the Czech Republic
(M\v SMT \v CR) under grant MSM 4781305904, and by Silesian University in Opava under grant IGS 2/2009.
The authors are pleased to thank M.~Kunzinger for stimulating discussions. A.S.\ gratefully acknowledges
the warm hospitality extended to him by the Department of Mathematics of the University of Vienna
during his visits in the course of preparation of the present paper.

It is our great pleasure to thank the referees for useful suggestions that have considerably improved the paper.

\end{document}